\documentclass[11pt]{article}

\usepackage[utf8]{inputenc}
\usepackage[T1]{fontenc}
\usepackage{fullpage}
\usepackage{url}
\usepackage{booktabs}
\usepackage{amsfonts}
\usepackage{amsmath}
\usepackage{amssymb}
\usepackage{amsthm}
\usepackage{nicefrac}
\usepackage{microtype}
\usepackage{graphicx}
\usepackage{natbib}
\usepackage{doi}
\usepackage{comment}
\usepackage{xcolor}
\usepackage{hyperref}

\usepackage{bm}

\usepackage{amsmath,amssymb,mathtools,centernot}
\newcommand\independent{\protect\mathpalette{\protect\independenT}{\perp}}
\def\independenT#1#2{\mathrel{\rlap{$#1#2$}\mkern3mu{#1#2}}}
\usepackage{caption}
\usepackage{float}
\usepackage{bm}
\usepackage{times}
\usepackage{calc}
\usepackage{graphicx}

\newtheorem{prop}{Proposition}
\usepackage{boondox-cal}
\usepackage{textgreek}
\usepackage{booktabs}
\usepackage{centernot}
\usepackage{amsthm}%
\usepackage{mathrsfs}%
\usepackage{enumitem}

\newtheorem{assumption}{Assumption}

\newlist{Properties}{enumerate}{2}
\setlist[Properties]{label=Property \arabic*., font=\textbf, itemindent=*}
\newenvironment{customprop}[1]
{\renewcommand\theprop{#1}\prop}
{\endprop}

\title{Examining the Efficacy of Coarsened Exact Matching as an Alternative to Propensity Score Matching}

\author{
  Fei Wan$^{1}$\\[1em]
  $^1$ Division of Public Health Sciences,Washington University in St. Louis, MO,USA \\[0.5em]
  \texttt{wan.fei@wustl.edu} \\[0.5em]
}
\date{}

\begin{document}

\maketitle

\begin{abstract}
Coarsened exact matching (CEM) is often promoted as a superior alternative to propensity score matching (PSM) for addressing imbalance, model dependence, bias, and efficiency. However, this recommendation remains uncertain. First, CEM is commonly mischaracterized as exact matching, despite relying on coarsened rather than original variables. This inexactness in matching introduces residual confounding, which necessitates accurate modeling of the outcome-confounder relationship post-matching to mitigate bias, thereby increasing vulnerability to model misspecification. Second, prior studies overlook that any imbalance between treated and untreated subjects matched on the same propensity score is attributable to random variation. Thus, claims that CEM outperforms PSM in reducing imbalance are unfounded, particularly when using metrics like Mahalanobis distance, which do not account for chance imbalance in PSM. Our simulations show that PSM reduces imbalance more effectively than CEM when evaluated with multivariate standardized mean differences (SMD), and unadjusted analyses indicate greater bias with CEM. While adjusted analyses in both CEM with autocoarsening and PSM may perform similarly when matching on few variables, CEM suffers from the curse of dimensionality as the number of factors increases, resulting in substantial data loss and unstable estimates. Increasing the level of coarsening may mitigate data loss but exacerbates residual confounding and model dependence. In contrast, both analytical results and simulations demonstrate that PSM is more robust to model misspecification and thus less model-dependent. Therefore, CEM is not a viable alternative to PSM when matching on a large number of covariates.
\end{abstract}


\section{Introduction}
\label{sec:intro}

Randomized clinical trials (RCTs) are considered the gold standard for investigating causal effects in comparative effectiveness research due to their ability to equalize the distribution of both observed and unobserved confounders across comparison groups. However, limitations such as time constraints, high costs, and ethical concerns often necessitate the use of observational studies as alternatives. In these studies, confounding remains a significant challenge, leading to the adoption of matching designs to emulate key aspects of RCTs and ensure balanced observed confounders between treatment arms \citep{stuart_matching_2010,dehejia_causal_1999,dehejia_propensity_2002}.

Exact covariate matching is statistically ideal for balancing confounders and reducing model dependence and bias, but it becomes impractical with high-dimensional confounders. The introduction of the propensity score \citep{rosenbaum_central_1983} improved matching design in comparative effectiveness research by condensing high-dimensional confounders into a single scalar, thereby alleviating the curse of dimensionality. Propensity score matching (PSM) effectively mimics an RCT under the balancing score and strong ignorability properties. However, the validity of PSM has been widely questioned among applied researchers due to a counter-intuitive phenomenon where PSM can paradoxically increase imbalance, statistical bias, and model dependence as it approaches exact matching by progressively pruning the worst-matched pairs, contrary to its intended purpose \citep{king_why_2019}. 

Coarsened Exact Matching (CEM) has been introduced as a recent advancement, claiming superior statistical properties, such as improved covariance balance, over methods like PSM \citep{iacus_causal_2012}. It has been adopted in studies published in leading clinical journals as a superior alternative to PSM \citep{haider_comparative_2013,Muennig_2016,wharam_breast_2018,sceats_nonoperative_2019,liu_propensity-score_2023,george_outcomes_2024}. CEM involves dividing continuous variables into bins or grouping categorical variables, a process known as ``coarsening'', and matching subjects based on these coarsened variables, followed by analysis using the original variables. While some studies highlight CEM's advantages over PSM in reducing imbalance, statistical bias, and model dependence \citep{king_why_2019,iacus_causal_2012}, several overlooked issues arise: 
\\
\vspace{0.1cm}
{\it{First}}, Previous studies \citep{iacus_causal_2012, king_why_2019, ripollone_implications_2018} suggest that PSM is prone to the PSM paradox and performs worse than CEM in balancing matched confounders. However, this conclusion is debatable for several reasons: i) Despite its name, CEM is an inexact matching method because it matches on coarsened confounders rather than their original forms. The inherent imbalance in CEM varies with the level of coarsening.  In contrast, the balancing score property of propensity scores ensures equal distributions of matched confounders between groups. ii) These studies often overlook that any imbalance in an exactly matched PSM design arises by chance. Even when treated and untreated subjects have identical propensity scores, random mismatches in covariates can occur \citep{rosenbaum_design_2020,wan_psm_jtcvs_2025,wan_psm_surgery_2025}. Thus, positive and negative mismatches occur with similar frequency and these mismatches tend to cancel out, and as the number of matched pairs grows, between-group imbalance becomes negligible. iii) Commonly used metrics like the Mahalanobis distance used in prior analyses fail to account for chance imbalances and do not capture the balancing score property of the propensity score accurately \citep{wan_psmparadox_2025}.
\\
\vspace{0.1cm}
{\it{Secondly}}, The inherent imbalance in matched confounders under CEM can lead to residual confounding, the magnitude of which depends on the degree of coarsening. Because coarsened matching alone cannot fully eliminate confounding, CEM requires post-matching regression adjustment using the original confounders. Unbiased effect estimation therefore hinges on correctly specifying the outcome model, increasing the risk of bias when the model is misspecified. While it is true that imbalance and model dependence in CEM are determined ex ante by the chosen level of coarsening \cite{iacus_causal_2012}, this theoretical advantage is often impractical. In high-dimensional settings, fine coarsening (e.g., auto-coarsening) results in excessive data loss, leaving only very coarse groupings as a feasible option. Such coarse grouping can produce substantial residual confounding, thereby greatly increasing model dependence and the risk of misspecification bias. In contrast, exactly matched PSM permits unbiased estimation via a simple difference in means \citep{rosenbaum_central_1983}, without requiring outcome modeling. Consequently, claims that CEM consistently reduces model dependence and bias more effectively than PSM are questionable.
\\
\vspace{0.1cm}
{\it{Lastly}}, As a covariate matching method, CEM is vulnerable to the curse of dimensionality \citep{Black_2020, Ripollone_2019}. Ripollone et al. \citep{Ripollone_2019} showed that CEM often leads to high bias and low precision in risk ratio estimates due to substantial reductions in sample size and outcome counts, a form of sparse data bias, especially when matching on high-dimensional covariates. While increasing the level of coarsening can alleviate this issue, it may also exacerbate residual confounding and model dependence. In contrast, PSM avoids the curse of dimensionality but could be less efficient than covariate matching when both methods produce a similar number of matched pairs \citep{wan_psmparadox_2025}.

Although previous studies have identified challenges with CEM, particularly sparse data bias from auto-coarsening in high-dimensional settings \citep{ripollone_implications_2018, Black_2020}, they largely overlook the method’s inherent inexactness and resulting residual confounding. This gap in literature has contributed to the common misconception that CEM performs exact matching and has obscured the complexities of post-matching adjustment required to address residual confounding bias. Comparisons with PSM are often limited or methodologically flawed \citep{wan_psmparadox_2025}. A balanced evaluation of CEM and PSM in terms of precision and bias across varying covariate dimensions is needed to provide clearer methodological guidance. To address prevailing misconceptions and fill gaps in the literature, this study aims to comprehensively compare CEM and PSM in their ability to reduce covariate imbalance, model dependence, and statistical bias. We also demonstrate the robustness of PSM to model misspecification analytically, leveraging its foundational properties of balancing scores and ignorability.

\section{Method}
\label{sec2}

\subsection{Definitions, Assumptions, Causal effect estimand}
\label{sec21}

\subsubsection{Definitions and Assumptions}
Let $Y$ denote a continuous outcome variable, and let $W \in \{0,1\}$ represent the binary treatment status, where 1 indicates the active treatment of interest and 0 indicates the control condition. Additionally, let $\bm{X}$ be a vector of $p$ baseline confounding variables, and $\bm{C}$ be the corresponding $p$-dimensional vector of coarsened versions of $\bm{X}$. We posit the following assumptions.

\begin{assumption}{The Stable Unit Treatment Value Assumption (SUTVA):} SUTVA has two components:
	\begin{itemize}
		\item[(a)] The potential outcomes for a subject are unaffected by the treatment assignments of other subjects (No interference).
		\item[(b)] For each unit, there are no different versions of each treatment, which lead to different potential outcomes. In other words, the potential outcome under a subject’s actual treatment assignment precisely matches the subject’s observed outcome (Consistency).
	\end{itemize} 
	\label{assumption1}
\end{assumption} 
SUTVA (a) allows us to use $Y_i(1)$ and $Y_i(0)$ to represent the two potential outcomes for individual $i$. Here, $Y_i(1)$ signifies the outcome that would have been observed had that individual received the treatment of interest, while $Y_i(0)$ indicating the counterfactual outcome had the individual received the control. SUTVA (b) allows us to express observed outcome for individual $i$ in terms of potential outcomes as: $Y_i = W_iY_i(1) + (1 - W_i)Y_i(0)$.

\begin{assumption}{The Conditional Ignorable Treatment Assumption:}	the potential outcomes are independent of treatment assignment given $\bm{X}$
	\begin{align*}
		\{Y_i(1),Y_i(0)\} \independent W_i | \bm{X}_i
	\end{align*}
	\label{assumption2}
\end{assumption}  
The treatment status for individual $i$ can be considered as randomly assigned conditional on covariates $\bm{X}_i$. Assumption (\ref{assumption2}) is also referred to as the ``unconfoundedness'' assumption and implies that there are no unmeasured confounders and all the confounders are included in $\bm{X}$.

\begin{assumption}{Positivity:}	the conditional probability of being assigned to either treatment given $\bm{X}$ is bounded away from 0 and 1
	\begin{align*}
		0<\mathbb{P}(W_i=1|\bm{X}_i)<1	\, ,
	\end{align*}
	where $\mathbb{P}(W_i=1|\bm{X}_i)$ denotes the probability of being assigned to the treatment group given the observed covariates $\bm{X}$.
	\label{assumption3}
\end{assumption}

\begin{assumption}{Linear outcome model:} In presence of a homogeneous treatment effect, the conditional expectation of $Y$ is assumed to take the following form \citep{king_why_2019}:
	\begin{align}
		\mathbb{E}(Y_i|W_i,\bm{X}_i)=\beta_0+\beta_1 W_i+ g(\bm{X}_i)
		\label{eq_1}
	\end{align}
	Here, $g(\cdot)$ is some arbitrary function. When the effects of $\bm{X}_i$ on $Y_i$ are linear additive, $g(\bm{X}_i)=\bm{\beta}_2 \bm{X}_i$, where $\beta_2$ is $p\times 1$ vector of coefficients measuring the effects of $\bm{X}$ on $Y$.
	In presence of a heterogeneous treatment effect, we assume the following interaction model:
	\begin{align}
		\mathbb{E}(Y_i \mid W_i, \boldsymbol{\tilde{X}}_i) = \beta_0 + \beta_1 W_i + (W_i \boldsymbol{\tilde{X}}_i) \boldsymbol{\theta}^\top + g(\bm{X}_i),
		\label{eq_2}	
	\end{align}
	\label{assumption4} 
	where  $\boldsymbol{\tilde{X}}_i$ is a $1 \times K$ subset of $\boldsymbol{X}$, with $0 \leq K \leq p$. $\bm{\theta}_{1\times K}$ represents the $W$ by $\boldsymbol{\tilde{X}}_i$ interaction effects.
\end{assumption}  

\subsubsection{Estimands}

({\em Homogeneous treatment effect:}) Common causal quantities of interest include the population average treatment effect (PATE) and the population average treatment effect among the treated (PATT). Assuming a super-population with $N$ units and $N_1$ treated subject, PATE is defined as
\begin{align*}
	\tau_{PATE}&=\mathbb{E}(Y_i(1)-Y_i(0)) \\
	&=\frac{1}{N}\sum^N_{i=1} \left( Y_i(1)-Y_i(0) \right)
\end{align*}
and PATT is defined as
\begin{align*}
	\tau_{PATT}&=\mathbb{E}(Y_i(1)-Y_i(0)|W_i=1) \\
	&=\frac{1}{N_1}\sum^{N_1}_{i=1} \left( Y_i(1)-Y_i(0) \right)
\end{align*}
PSM targets for PATT. If the treatment effect is homogeneous (the same across all individuals), the PATT and PATE will be identical. In such cases, the treatment effect is constant across all subjects, and it does not matter whether you focus on the treated group or the entire population.

Under the assumptions (1-4), $\tau_{PATE}$ can be expressed as
\begin{align*}
	\tau_{PATE}&=\mathbb{E}_{\bm{X}}\mathbb{E}(Y_i(1)-Y_i(0)|\bm{X}) \\
	&=\mathbb{E}_{\bm{X}}(\mathbb{E}(Y_i|W_i=1,\bm{X})-\mathbb{E}(Y_i|W_i=0,\bm{X})) \\
	&=\beta_1
\end{align*}
Similarly, we can show that $\tau_{PATT}=\beta_1$. PATE=PATT because Assumption (\ref{assumption4}) implies that $\beta_1$ in model (\ref{eq_1}) represents a homogeneous treatment effect. 

Other studies \citep{king_why_2019,iacus_causal_2012} investigated the average treatment effect in the specific sample, rather than in the population. Let $S$ denote a sample of size $n$ with $n_1$ treated subjects, the sample-average treatment effect (SATE) can be defined as  
\begin{align*}
	\tau_{SATE}&=\mathbb{E}(Y_i(1)-Y_i(0)|S_i=1) \\
	&=\frac{1}{n}\sum^n_{i=1} \big ( Y_i(1)-Y_i(0) \big)
\end{align*}

Similarly, sample treatment effect among the treated (SATT) is defined as
\begin{align*}
	\tau_{SATT}&=\mathbb{E}(Y_i(1)-Y_i(0)|W_i=1,S_i=1) \\
	&=\frac{1}{n_1}\sum^{n_1}_{i=1} W_i \big ( Y_i(1)-Y_i(0) \big)
\end{align*} 

Under assumption (\ref{assumption4}), we have $\tau_{PATE} = \tau_{SATE} = \tau_{PATT} = \tau_{SATT} = \beta_1$. In the absence of treatment effect heterogeneity, all these quantities coincide and equal the constant treatment effect. 

({\em Heterogeneous treatment effect}:) In the presence of heterogeneous treatment effects, however, the values of SATE and SATT depend on the specific units included in the sample and may vary across different samples. Regardless of whether the current sample is drawn from a super-population, SATE and SATT reflect the average treatment effect within the sample itself. Similarly, with treatment effect heterogeneity, PATE and PATT can also differ. Specifically,

At the individual level, we define:
\begin{align*}
	\tau_i = Y_i(1) - Y_i(0)
\end{align*}

At the subgroup (covariate) level, the conditional average treatment effect (CATE) is:
\begin{align*}
	\tau(x) = \mathbb{E}(\tau_i \mid X_i = x)
\end{align*}

Under the interaction model (\ref{eq_2}), the individual CATE is:
\begin{align*}
	\tau_i=\tau(\bm{\tilde{X}}_i) = \beta_1 +  \boldsymbol{\tilde{X}}_i \bm{\theta}^\top
\end{align*}

In this case, the individual treatment effect $\tau_i$ is constant among individuals with the same value of $\bm{\tilde{X}}=\bm{\tilde{x}}$ but varies across different values of $X_{1i}$. It follows:
\begin{itemize}
	\item[(1)] PATE: $\tau_{PATE}=\mathbb{E}(\tau_i)=\beta_1+ \mathbb{E}(\boldsymbol{\tilde{X}}_i)\bm{\theta}^\top$. When $\mathbb{E}(\boldsymbol{\tilde{X}}_i) =\bm{0}$, PATE simplifies to $\tau_{PATE} = \beta_1$.
	\item[(2)] PATT: $\tau_{PATT}=\mathbb{E}(\tau_i|W_i=1)=\beta_1+ \mathbb{E}(\boldsymbol{\tilde{X}}_i|W_i=1)\bm{\theta}^\top$. 
\end{itemize}
PATE and PATT differ unless $\mathbb{E}(\bm{\tilde{X}}_i|W_i=1)=\mathbb{E}(\bm{\tilde{X}}_i)$. That is, PATE = PATT if and only if the distribution of $\bm{\tilde{X}}_i$ is the same in the treated and the overall population. But in observational data, this is usually not true — people who receive treatment often differ systematically from those who don't. For example, more severe patients might be more likely to receive treatment, shifting 
$\mathbb{E}(\bm{\tilde{X}}_i|W_i=1)$ away from $\mathbb{E}(\bm{\tilde{X}}_i)$.

This study focuses on the population effect for several reasons:
i) PATT aligns with common research goals in clinical studies, where researchers aim to generalize findings to the target population \cite{austin_methods_2009, austin_comparison_2014}.
ii) PATT is widely used in methodological research evaluating propensity score methods \cite{austin_optimal_2011, wan_evaluation_2018, wan_matched_2019}.
iii) In PSM designs, imbalance occurs by chance, and interpreting such chance imbalances is more intuitive for PATT.

\subsection{Propensity score matching}

The propensity score is a conditional probability of receiving the treatment rather than the control given the observed covariates, denoted by:
\begin{align*}
	e(\bm{X})=P(W=1|\bm{X})
\end{align*}
The propensity score normally is not known in practice and is often estimated using the following logit model:
\begin{align}
	\text{logit}(W_i=1|\bm{X}_i)=\alpha_0+ \phi(\bm{X}_i),
	\label{eq_3}
\end{align}
where $\alpha_0$ represents the intercept and $\phi(\cdot)$ is an arbitrary function. When the effects of $\bm{X}_i$ on $W_i$ are linear additive, $\phi(\bm{X}_i)=\bm{\alpha}_1\bm{X}_i$, where $\bm{\alpha}_1$ is $p\times 1$ vector of regression coefficients.  The propensity score possesses the following two properties:

\begin{Properties}
	\item (\textbf{Balancing score}) The propensity score $e(\bm{X})$ balances the distribution of $\bm{X}$ between the comparison groups:
	\begin{align*}
		W \independent \bm{X} | e(\bm{X})
	\end{align*}
	
	\item (\textbf{Strongly ignorable treatment assignment (SITA)}) If $W$ is unconfounded given $\bm{X}$, then $W$ is unconfounded given $e(\bm{X})$. Formally,	
	\begin{align*}
		\{ Y(1),Y(0) \} \independent W | \bm{X} \implies	\{ Y(1),Y(0) \} \independent W| e(\bm{X})
	\end{align*}
\end{Properties}

When all confounders are observed and included in $\bm{X}$, the balancing score and strong ignorability properties of the propensity score ensure that treatment assignment is random among subjects with the same score. However, even when matching on the same propensity score, treated and untreated subjects within matched sets often show distinct covariate values, a characteristic that sets PSM apart from other covariate matching designs. These within-pair imbalances occur by chance \citep{rosenbaum_design_2020}, with positive and negative differences equally likely. As the number of matched pairs increases, the average imbalance converges toward zero, balancing $\bm{X}$ between treated and untreated groups in PSM. Thus, chance imbalances do not introduce residual confounding or bias, ensuring Property (2) holds. 
This feature can also be explained within a regression framework \citep{wan_interpretation_2021}. Instead of directly including confounders as covariates in the regression model, adjusting for the propensity score decomposes confounders into two components: the propensity score and a residual term representing the covariate differences among subjects having the same propensity scores. By conditioning on the propensity score, the residual becomes independent of treatment assignment and is treated as random noise, which does not bias treatment effect estimation.

Chance imbalance is inevitable, even in randomized designs. Senn \citep{senn_seven_2013} emphasized that randomization does not ensure balance in a single completely randomized design. While blocking can address some imbalance, unblocked variables may still differ between arms. However, chance imbalance does not bias point estimates, as any deviation from the true value reflects variability, not statistical bias. Incorporating baseline variables in an analysis of covariance (ANCOVA) can further enhance precision \cite{wan_analyzing_2019,wan_analyzing_2020,wan_analyzing_2021}. 

\subsection{Coarsened exact matching}
\label{sec23}

The central concept of CEM involves grouping values that are substantively indistinguishable for each matched variable and assigning them the same numerical value to coarsen the data. Following this, treated and untreated subjects are matched based on the exact values of these coarsened variables. After this matching process is finished, the coarsened data is discarded, while the original (uncoarsened) values of each matched variable are preserved in the matched data for subsequent post-matching regression analysis. The specific steps include:

\begin{itemize}
	\item[(i)]
	The ``coarsening'' or binning step involves creating coarsened versions, denoted as $\bm{C}$, of matched variables $\bm{X$} by dividing each matched variable $X_j$ into bins. Each bin represents a distinct value of the coarsened factor $C_j$. Coarsening process can be performed through explicit user choices based on subject knowledge or by employing automatic coarsening algorithms. For instance, an ordinal 7-point Likert scale can be grouped into a 3-point scale: \{completely disagree, strongly disagree, disagree\}, \{neutral\}, \{agree, strongly agree, completely agree\}. Continuous years of education can be categorized into different levels of schooling. In cases where researchers lack subject knowledge, automatic matching algorithms can be utilized. The default ``Sturges'' option in the $\bm{\texttt{MatchIt}}$ package ($\texttt{R}$) uses the formulae $\texttt{ceiling}(\text{log}_2(n) + 1)$ to determine the number of bins, where $n$ represents the sample size. For example, with a sample size of 300 observations, each continuous matching variable will generate 10 bins.
	\item[(ii)] The ``matching'' step involves pairing treated and control subjects based on the exact values of coarsened matched variables $\bm{C}$. Many-to-many matching ($K:M$) can be performed, grouping treated and control units with the same $\bm{C}$ values into multidimensional strata. In this context, CEM serves as a stratification method. Within the $s^{th}$ stratum ($s=1,2,\cdots,S$, and $S$ is the total number of matching strata), case subjects receive a weight of 1, while control subjects receive a weight of $\frac{m_C}{m_T}\frac{m^s_T}{m^s_C}$, where $m_C$ and $m_T$ are the total numbers of control and treated subjects, and $m^s_C$ and $m^s_T$ are the numbers of control and treated subjects in the $s^{th}$ stratum. Alternatively, 1:1 matching can be performed within matched strata to avoid using weights. The $\bm{\texttt{MatchIt}}$ package employs nearest neighbor matching without replacement within each stratum, ensuring an equal number of treated and control units in all strata.
	\item[(iii)] After preprocessing data with CEM, post-matching analysis is necessary to estimate the treatment effect. Blackwell et al. \cite{blackwell_cem:_2009} suggest that the analyst may use a simple sample mean difference or any statistical model that would have been applied without matching. However, as discussed later, unadjusted analyses are generally biased due to residual confounding inherent in CEM. To achieve an unbiased estimation of the treatment effect, it is essential to use a correctly specified regression model.
\end{itemize}

CEM does not provide exact matches on the original variables $\bm{X}$ and therefore lacks two key properties of the propensity score. First, exact matching on coarsened variables $\bm{C}$ does not guarantee balanced distributions of confounders $\bm{X}$ across comparison groups. That is, $\bm{X} \not\!\perp\!\!\!\perp W | \bm{C}$  because $\mathbb{P}(W|\bm{X}, \bm{C}) = \mathbb{P}(W| \bm{X})$ and $\mathbb{P}(W|\bm{X}, \bm{C}) \neq \mathbb{P}(W| \bm{C})$. Second, the standard ignorability assumption ${Y(1),Y(0)} \independent W | \bm{X}$ does not imply ${Y(1),Y(0)} \independent W | \bm{C}$. Consequently, residual confounding is expected when using CEM, and its magnitude depends on the degree of coarsening in $\bm{C}$. For a given coarsened stratum $\bm{C} = \bm{c}$, define the within-bin covariate imbalance as
\begin{align*}
	\delta_c = \mathbb{E}(\bm{X} \mid W = 1, \bm{C} = \bm{c}) - \mathbb{E}(\bm{X} \mid W = 0, \bm{C} = \bm{c}).
\end{align*}

As the sample size increases, the empirical mean difference converges almost surely to $\delta_c$ by the law of large numbers:
\begin{align*}
	\frac{1}{n_{1c}} \sum_{i: W_i = 1,\, \bm{C}_i = \bm{c}} \bm{X}_i - \frac{1}{n_{0c}} \sum_{i: W_i = 0,\, \bm{C}_i = \bm{c}} \bm{X}_i \xrightarrow{a.s.} \delta_c.
\end{align*}
In theory, $\delta_c = 0$ if and only if $W \perp\!\!\!\perp \bm{X} \mid \bm{C}$, meaning the coarsening is fine enough to fully block the association between treatment assignment $W$ and the original covariates $\bm{X}$. However, because coarsening is a lossy transformation, this condition rarely holds in practice. Consequently, the imbalance in $\bm{X}$ within each matched stratum is not due to random variation but is systematic, driven by residual dependence between $W$ and $\bm{X}$ within levels of $\bm{C}$. This imbalance persists regardless of sample size.

CEM is often motivated as a strategy to reduce the curse of dimensionality by matching on a lower-dimensional, coarsened version of $\bm{X}$. However, its logic becomes circular: we reduce dimensionality by coarsening, knowing coarsening loses information, yet must assume that the coarsened variables $\bm{C}$ preserve all information in $\bm{X}$ relevant to $W$ in order to justify the method’s validity \citep{iacus_causal_2012}. In reality, unless $\bm{C}$ is a sufficient statistic for $\bm{X}$ with respect to $W$, which is unlikely without exact matching, CEM introduces structural imbalance and residual confounding that do not diminish with increasing sample size.

\subsection{Imbalance Metric}
\label{sec22}

Some commonly used multivariate imbalance metrics include the Mahalanobis distance, $L_1$ imbalance metrics, and absolute between group mean difference. These metrics are defined as follows:
\begin{itemize}
	\item[(i)] The averaged pairwise Mahalanobis distance \citep{king_why_2019, King2011}:
	\begin{align*}
		I_1(\bm{X})=\text{mean}_{i\in \{i\}}d(\bm{X}_i,\bm{X}_{j(i)})
	\end{align*} 
	Here $I_1(\bm{X})$ represents the average pairwise distance between treated subject $i$ with covariates $\bm{X_i}$, and the closest untreated subject $j$ with covariates $\bm{X_{j(i)}}$. $d(\cdot)$ is the Mahalanobis distance:
	\begin{align*}
		d(\bm{X}_i,\bm{X}_{j(i)})=\sqrt{(\bm{X}_i-\bm{X}_{j(i)})^{'}\Sigma^{-1}(\bm{X}_i-\bm{X}_{j(i)})}
	\end{align*}
	where $\Sigma$  represents the sample
	covariance matrix of the original data. 
	\item[(ii)] The between-group  Mahalanobis distance \cite{ripollone_implications_2018}:
	\begin{align*}
		I_2(\bar{\bm{X}}_1,\bar{\bm{X}}_0)=\sqrt{(\bar{\bm{X}}_1-\bar{\bm{X}}_{0})^{'}\Sigma^{-1}(\bar{\bm{X}}_1-\bar{\bm{X}}_{0})}
	\end{align*}
	Here $\bar{\bm{X}}_1$ and $\bar{\bm{X}}_0$ are the vectors of covariate means in the treated and untreated groups. 
	\item[(iii)] The absolute group-mean-difference based imbalance metric \cite{King2011}:
	\begin{align*}
		I_3(\bar{\bm{X}}_1,\bar{\bm{X}}_0)=\sum_{j=1}^J |d(j)|	
	\end{align*} 
	Here $d(j)=\bar{X}_{1j}-\bar{X}_{0j}$ is the group mean difference for the $j$th confounder. $I_3$ is the sum of the absolute mean differences for all $J$ confounders.
	\item[(iv)] The $L_1$ imbalance metric \cite{King2011}:
	\begin{align*}
		I_4(H)=\sum_{(\ell_1\cdots\ell_k)\in H} |f_{\ell_1\cdots\ell_k}-g(\ell_1\cdots\ell_k)|	
	\end{align*} 
	$H$ is a chosen set of bin sizes. $f_{\ell_1\cdots\ell_k}$  is the relative empirical frequency of treated subjects in bin with coordinates $\ell_1\cdots\ell_k$, and similarly for $g(\ell_1\cdots\ell_k)$ among control units. $I_4(H)$ is the difference of the multivariate histogram of the treated units and the multivariate histogram of the control units.
\end{itemize}

Prior studies \citep{king_why_2019,ripollone_implications_2018,Ripollone_2019,King2011} used the imbalance metrics listed above to demonstrate a persistent worsening of imbalance in PSM, as worst-matched pairs were removed and PSM approached exact matching in simulation studies and in single data sets. However, there are several issues in these findings:
\begin{itemize}
	\item[(i)] In simulation studies, matched data is created from each random sample of finite size. Due to chance imbalances, between-group differences in confounders can be positive in some matched samples but negative in others. However, these differences average out to zero across all matched samples. Imbalance metrics $I_1-I_4$ ignore the direction of chance imbalance, always returning absolute values, which cannot average to zero.
	
	\item[(ii)] It is improbable to demonstrate chance imbalance in a single dataset, particularly as sample sizes decrease due to progressive pruning, since balancing is a property of large samples. In a single matched dataset, chance imbalance implies that differences between treated and untreated subjects may be positive in some matched pairs and negative in others. As the number of matched pairs increases (rather than decreasing due to progressive pruning of worst-matched pairs), the average between-group imbalance tends to approach zero. Consequently, the worsening imbalance in PSM observed during continual pruning \cite{king_why_2019}, as assessed using imbalance metrics $I_1-I_4$, reflects a higher likelihood of significant chance imbalance as sample sizes shrink, rather than an increase in systematic imbalance that could introduce confounding bias. A similar phenomenon is observed in randomized trials, where significant baseline covariate imbalances are more likely in smaller studies compared to larger ones.

\end{itemize}

To accurately measure the overall between-group imbalance across multiple matched variables, while accounting for the direction of chance imbalances, we propose a multivariate imbalance metric based on the standardized mean difference (SMD) for use in the simulation study:

\begin{align}
	I_5(\bar{\bm{X}}_1,\bar{\bm{X}}_0)=\sum^J_{j=1} \left |\frac{\sum^K_{k=1} \text{SMD}_{j}^{(k)}}{K} \right | ,
	\label{eq_4}
\end{align}
and
\begin{align*}
	\text{SMD}_j^{(k)}=	\frac{(\bar{X}^{(k)}_{1j}-\bar{X}^{(k)}_{0j})}{\sqrt{\frac{(S^2_{1j})^{(k)}+(S^2_{0j})^{(k)}}{2}}}
\end{align*}
Here, $\bar{X}^{(k)}_{1j}$, $\bar{X}^{(k)}_{0j}$, $(S^2_{1j})^{(k)}$, and $(S^2_{0j})^{(k)}$ represent sample means and variances of the $j^{th}$ confounder in the treated and untreated groups from the $k^{th}$ simulated sample. SMD is a commonly used metric in balance diagnosis and captures both the size and direction of imbalance. Formula (\ref{eq_4}) involves two steps: 
\begin{itemize}
	\item[(i)]  For each matching variable $X_j$, where $j=1,2,\cdots,J$, we compute the SMD of $X_j$ between the treated and untreated groups in the $k^{th}$ sample, where $k=1,2,\cdots, K$. We then sum up $K$ SMDs across all $K$ simulated samples and compute the averaged SMD for $X_j$.  
	\item[(ii)] Next we sum up the $J$ averaged SMDs in their absolute values with equal weighting over all confounders. Equal weighting is justifiable because the SMDs for all the confounders are normalized to have zero mean and unit standard deviation. 
\end{itemize}
If the covariate difference between treated and untreated subjects within matched pairs occurs by chance, the SMD of a confounder may be negative in some simulated samples and positive in others. However, these SMDs would average to zero with a large number of simulated samples (i.e., $\displaystyle  \lim_{K \to \infty}\frac{\sum^K_{k=1} \text{SMD}^{(k)}}{K} =0$). $I_5$ is specifically crafted to assess the chance imbalance in simulation studies, in which different random samples are drawn. If within-pair covariate difference occurs not by chance, $I_5$ would generally deviate from zero.

\subsection{Model Dependence and Estimation Bias}

Model dependence refers to the sensitivity of treatment effect estimation to the statistical models employed. A significant drawback of using regression to control for confounding is its reliance on correctly modeling the outcome-confounder relationship, making it susceptible to model misspecification bias. Incorrectly specifying this relationship, such as by ignoring non-linearities or interaction terms, can introduce bias into the estimation of the treatment effect. Conversely, a good matching design acts as a non-parametric preprocessing tool that reduces model dependence in subsequent analysis, offering a distinct advantage over regression adjustment \citep{ho_matching_2007}. When matching is exact, confounders are evenly distributed between comparison groups, allowing a simple mean difference to provide an unbiased estimate of the conditional treatment effect, similar to analyzing an RCT \citep{ho_matching_2007,guo_statistical_2023}. Even with imperfect matching, applying a misspecified regression model to matched data is less biased than when applied to the original unmatched data \citep{guo_statistical_2023}. These findings align with the earlier assertion that a well-designed matching process, coupled with subsequent regression adjustment, generally yields the least biased estimates \citep{rubin_use_1973}. The reduction in model dependence through matching can be attributed to its ability to balance any function of $\bm{X}$, thereby rendering $W$ and any function of $\bm{X}$ approximately orthogonal in the matched sample. Consequently, according to least-squares theory, the inclusion or exclusion of a nearly orthogonal predictor has negligible effects on other regression coefficients \cite{guo_statistical_2023}.

However, there is still no consensus on the appropriate metric for measuring model dependence in a design. King and Nelsen \citep{king_why_2019} compared model dependence and statistical bias across various matching designs. They assumed researchers, uncertain about the correct model, would explore many different regression models, differing in sets and functional forms of $\bm{X}$, during post-matching analysis and then select the most favorable estimate-a process commonly known as ``cherry-picking.''. In their demonstration, they tested 512 models for two matched variables. Model dependence was measured by the variance in treatment effect estimates from this cherry-picking process used to analyze a matched dataset, while bias was defined as the difference between the maximum estimate and the true causal parameter. Formally, let $\boldsymbol{\hat{\beta}_1}=\langle \hat{\beta}_{1,\mathcal{M}_1}, \hat{\beta}_{1,\mathcal{M}_2}, \cdots, \hat{\beta}_{1,\mathcal{M}_K} \rangle$ represent a $1 \times K$ vector of estimates for $\beta_1$ from $K$ models, with $\hat{\beta}_{1,\mathcal{M}_k}$ denoting the estimate from model $\mathcal{M}_k$ ($k=1, 2, \dots, K$). Model dependence, measured as the variance in these estimates, is given by $\sigma_\mathcal{M}^2=\text{Var}\langle \hat{\beta}_{1,\mathcal{M}_1}, \hat{\beta}_{1,\mathcal{M}_2}, \cdots, \hat{\beta}_{1,\mathcal{M}_K}\rangle$ and estimated by
\begin{align*}
	\hat{\sigma}_\mathcal{M}^2 = \frac{\sum^K_{k=1} (\hat{\beta}_{1,\mathcal{M}_k} - \bar{\hat{\boldsymbol{\beta}}}_1)}{K-1}, 
\end{align*}
where $\bar{\hat{\boldsymbol{\beta}}}_1=\frac{\sum^K_{k=1} \hat{\beta}_{1,\mathcal{M}_k}}{K}$. This variance metric differs from commonly used metrics like mean square error (MSE) for a single parameter, which is defined as follows:
\begin{align*}
	\text{MSE}(\hat{\beta}_{1,\mathcal{M}_k}) =\mathbb{E}(\hat{\beta}_{1,\mathcal{M}_k}-\beta_1)^2 =\text{Var}(\hat{\beta}_{1,\mathcal{M}_k})+\text{Bias}(\hat{\beta}_{1,\mathcal{M}_k})^2	
\end{align*}
Notably, $\sigma_\mathcal{M}^2$ does not account for the bias component because it does not include the true causal parameter. Since researchers are suggested to select the largest estimate for the treatment effect, we denote it by $\hat{\beta}_{1,max}=\text{Max}\langle \hat{\beta}_{1,\mathcal{M}_1}, \hat{\beta}_{1,\mathcal{M}_2}, \cdots, \hat{\beta}_{1,\mathcal{M}_K}\rangle$. 

King and Nelsen \citep{King2011} found that as PSM approaches exact matching through continued pruning of the relatively worst-matched pairs, it paradoxically increases imbalance, model dependence, and statistical bias, whereas other covariate-matching methods avoid these issues. However, measuring model dependence and statistical bias based on a cherry-picking process is problematic for several reasons:
\\
\vspace{0.1cm}
{\it{First}}, the metric $\sigma_\mathcal{M}^2$ may conflate model dependence with design efficiency and sampling variability of effect estimators. Using $\sigma_\mathcal{M}^2$ to quantify model dependence likely stems from the idea that estimates in a less model-dependent design should be closer to each other. However, in a more efficient design, estimates from different models may also exhibit less volatility and thus appear closer. PSM is typically less efficient than other covariate-matching methods, as it matches on a summary score rather than directly on covariates. Nevertheless, it offers a practical solution to the curse of dimensionality, a major limitation of covariate-matching designs. Additionally, $\sigma_\mathcal{M}^2$ can be influenced by the variance estimator of $\hat{\beta}_{1,\mathcal{M}k}, k=1,2\cdots,K$, which is inherently model-specific and also depends on sample size. As pruning continues in a shrinking dataset, the variability of individual estimates increases, leading to a rise in $\sigma_\mathcal{M}^2$.  Lastly, $\sigma_\mathcal{M}^2$ does not account for the bias component. While CEM may yield less variable estimates than PSM, as PSM is generally less efficient than covariate-matching designs, unbiased estimation in CEM still relies on correctly modeling the outcome-confounder relationship due to its inexact matching. Therefore, it is more appropriate to define the model dependence of a matching design strictly as the reliance on correct modeling for unbiased estimation in post-matching analysis, distinguishing it as a property separate from design efficiency.\\
\vspace{0.1cm}
{\it{Secondly}}, robustness against cherry-picking should not be a criterion for validating a matching design, as this practice can yield extreme effect estimates even in randomized studies \citep{tsiatis_covariate_2008}. The theory of order statistics shows that $\hat{\beta}_{1,max}$ is biased, even when individual estimates are unbiased. Design efficiency also affects $\hat{\beta}_{1,max}$, as more efficient designs produce tighter estimates with lower variance, potentially making the maximum less extreme. Moreover, $\hat{\beta}_{1,max}$ depends on sample size. As pruning continues in a shrinking dataset, increasing variability among individual estimates leads to a more extreme $\hat{\beta}_{1,max}$.
\\
\vspace{0.1cm}
{\it{Finally}}, when researchers engage in cherry-picking analysis, they may lose sight of the purpose of matching designs and the rationale behind using a simple group mean difference—an unadjusted analysis and, by default, a misspecified model that completely ignores the modeling of matched confounders—as a valid, unbiased effect estimator\citep{rosenbaum_central_1983}.  The ultimate goal of designing a matched comparative effectiveness study is to identify an unbiased estimate of treatment effect,not searching for some ``best'' effect estimate.  Matching designs require navigating complex algorithms and often discarding valuable unmatched data. If, after all these efforts, researchers still struggle to select the correct outcome model from numerous candidates in post-matching analysis, they should reconsider whether matching offers any advantage over regression adjustment in the original unmatched data.  

A more appropriate definition of model dependence for a matching design is whether unbiased effect estimation requires a correctly specified model. Instead of assessing the similarity of treatment effect estimates from different models within a single matched dataset, we should assess whether certain misspecified models in the evaluated matching design can still produce unbiased results. To assess model dependence across different designs and avoid conflating it with design efficiency, we will compare the biases arising from commonly used misspecified models applied to both designs. A good matching design should reduce sensitivity to model dependence, resulting in less bias from misspecified models compared to unmatched data. If one design mitigates model dependence more effectively, the same misspecified model should yield a less biased estimate in that design. We hypothesize that PSM will be more robust than CEM, as the residual confounding in CEM’s coarsening requires accurate modeling of the outcome-confounder relationship, increasing model dependence and the risk of bias. In contrast, we prove that a class of misspecified regression models with any combination of confounders in their linear terms can asymptotically yield unbiased ATT estimates in an exactly matched PSM design, as demonstrated in the following proposition:

\begin{prop}
	Suppose the true outcome model is defined in equation (\ref{eq_1}) and we fit a misspecified linear regression model in an exactly matched PSM design to estimate ATT, defined as follows:
	\begin{align*}
		\mathbb{E}(Y_i|\boldsymbol{\tilde{X}_i})=\gamma_0+\gamma_1 W_i +  \boldsymbol{\tilde{X}_i} \boldsymbol{\gamma_{2}}^\top \, , \,  i=1,2,\cdots, N,
	\end{align*}
	where $\boldsymbol{\tilde{X}}_i$ is a $1 \times K$ subset of $\boldsymbol{X}$, with $0 \leq K \leq p$. $\boldsymbol{\tilde{X}}_i$ may be an empty set or include any combination of confounders $\boldsymbol{X}$. $\boldsymbol{\gamma_{2}}$ is $1 \times K$ vector of regression coefficients. Here, $N$ represents the sample size of the PSM design. The resulting OLS estimator satisfies
	\begin{align*}
		\hat{\gamma}_1 \overset{p}{\to} \mathbb{E}(Y(1)-Y(0)|W=1) ,
	\end{align*} 
\end{prop}

\begin{proof}
	Details in the Appendix.	
\end{proof}	
Thus, in an exactly matched PSM design, unbiased effect estimation does not require the correct model (\ref{eq_1}). An unadjusted model (e.g., sample mean difference) or a linear model including any matched confounders in linear terms suffices, without concern for nonlinear or interaction terms. This reduction in model dependence is a desirable property of an effective matching design, compensating for the complexity of matching algorithms and associated data loss. Notably, not all commonly used matching designs can reduce model dependence. For instance, matching in case-control studies may not eliminate all backdoor paths between the exposure and outcome, and it can also complicate the outcome-covariate relationship, increasing model dependence in post-matching analysis \citep{wan_mat_2021,wan_conditional_2022,wan_does_2024}. In the following simulation study, we will analyze both PSM and CEM designs using an unadjusted model, $\mathcal{M}(W)$, and an adjusted model with linear terms of $\bm{X}$, $\mathcal{M}(W, \bm{X})$. Additionally, we will apply $\mathcal{M}(W)$ and $\mathcal{M}(W, \bm{X})$ to the original unmatched data.

\subsection{Relative Efficiency of PSM vs. Covariate Matching Design}

Consider a 1:1 exact covariate matching on confounders $\boldsymbol{\mathbf{X}}$, forming $M$ matched pairs. In comparison, a 1:1 matching design on exact logit propensity scores forms $N$ matched pairs. The relative efficiency of the treatment effect estimator $\bar{Y}_{1.} - \bar{Y}_{0.}$ under PSM compared to covariate matching is given by \cite{wan_psmparadox_2025}:  
\begin{align*}
	\text{Relative Efficiency of PSM vs CEM} = \frac{\sigma^2_{\epsilon}}{\sigma^2_{\nu} + \sigma^2_{\epsilon}} \frac{n}{m},
\end{align*}
where $\sigma^2_{\nu} = \sin^2(\theta) \, \boldsymbol{\beta_2}^T \tilde{\Sigma}_{\mathbf{X}} \boldsymbol{\beta_2}$, and $\tilde{\Sigma}_{\mathbf{X}}$ is the $p \times p$ variance-covariance matrix of $\mathbf{X}$ in the matched population under PSM. $\theta$ is the angle between two coefficient vectors $\boldsymbol{\alpha_1}$ and$\boldsymbol{\beta_2}$. $\sigma^2_{\epsilon}$ represents the error variance for model (\ref{eq_1}).  When both designs have similar number of matched pairs ($n\approxeq m$),  $\frac{\sigma^2_{\epsilon}}{\sigma^2_{\nu} + \sigma^2_{\epsilon}} < 1$. In this case, PSM is less efficient because it uses only partial information from $\boldsymbol{X}$, specifically $e(\boldsymbol{\mathbf{X}})$ as the blocking factor. However, when the number of matching variables is large, covariate matching encounters the curse of dimensionality, leading to significantly fewer matched pairs. When $m \ll n$, the variance reduction achieved by covariate matching may not offset the data loss. In this case, PSM becomes more efficient.

As a covariate matching method, CEM may outperform PSM in efficiency when the number of matching factors is limited, and the number of matched pairs is comparable to PSM. However, unbiased effect estimation in CEM depends more heavily on accurate modeling. The choice between CEM and PSM involves a trade off between statistical bias and precision, which can be evaluated using MSE in the simulation study.

\section{Simulation}
\label{sec5}
In this simulation study, we aim to demonstrate the following key points: 
\begin{itemize}
	\item[(i)] {\em Systematic Imbalance in CEM}: Unlike PSM, imbalance in CEM is systematic and does not diminish with increasing sample size. As a result, CEM is more prone to residual confounding bias, whereas PSM—when using a reasonable caliper—better approximates a randomized design. The unadjusted analysis $\mathcal{M}(W)$ is substantially less biased under PSM than CEM.
	\item[(ii)] {\em Reduction in Model Dependence}: Both CEM (``Auto'') and PSM reduce model dependence. When applying the same misspecified model $\mathcal{M}(W,\bm{X})$, bias is significantly lower with either method compared to the unmatched data.
	\item [(iii)] {\em Curse of Dimensionality in CEM:} As the number of matched variables increases, CEM suffers from the curse of dimensionality, leading to substantial data loss and unstable estimates. PSM avoids this issue.
	\item [(iv)] {\em Trade-off with Coarsening Level:} Increasing the coarsening level in CEM can partially mitigate data loss but at the cost of higher imbalance, greater statistical bias, and increased model dependence.
	\item [(V)]{\em Heterogeneous treatment effect}: We aim to assess whether similar patterns persist in the presence of heterogeneous treatment effects and to investigate how the proportion of treated subjects and model misspecification may influence statistical bias.
\end{itemize}

\subsection{Simulation design}

\label{sec31}

\textbf{\textit{The general settings:}} We follow these steps to generate $Y$, $W$, and $\bm{X}$ for the simulation data: 
\begin{itemize}
	\item[(i)] We initially generated $p$ independent confounding variables $X_1, X_2, \ldots, X_p$, $\text{where} \, p=2,5,7,$  each drawn from a normal distribution $N(0,10)$.
	
	\item[(ii)] To generate the outcome $Y$ and treatment status $W$, we began by forming two coefficient vectors, $\bm{\beta}_2$ and $\bm{\alpha}_1$, for equations (\ref{eq_1}) and (\ref{eq_2}). For $\bm{\beta}_2$, we initially sampled integers randomly from the range of 1 to 9. Subsequently, we normalized this coefficient vector to become a unit vector, with the sign of each element determined using a Bernoulli(0.5). Finally, we set $\bm{\beta}_2$ equal to $k$ times its normalized factor, where $k=1.2$. The same procedure was then repeated to generate $\bm{\alpha}_1$, with $\bm{\alpha}_1$ set to 1 multiplied by the normalized vector. We generated 300 pairs of $\bm{\beta}_2$ and $\bm{\alpha}_1$ with pairwise $sine$ distances approximately evenly distributed over [0, 1]. The $sine$ distance serves as a measure of dissimilarity between two vectors, spanning the range from 0 to 1, where a larger value indicates a greater dissimilarity between the two vectors \cite{wan_interpretation_2021}. We selected various pairs of $\bm{\alpha}_1$ and $\bm{\beta}_2$ with differing sine distances to ensure generalizability, as prior studies typically relied on specific regression coefficient sets, such as proportional coefficients \cite{wan_cautionary_2024}.
	
	\item[(iii)] $W$ was generated using the treatment model (\ref{eq_2}) with the intercept $\alpha_0$ set to -0.9, resulting in approximately 30\% of simulated subjects receiving the treatment. Following this, the outcome variable $Y$ was generated using the linear outcome model (\ref{eq_1}) or (\ref{eq_2}), incorporating an error term $\sim N(0, 1)$. 
\end{itemize}

After generating each simulated data set with a sample size of $5000$, we employed CEM and PSM to address confounding bias. For PSM, we computed the the propensity score for each subject through a correctly specified logistic regression model. Subsequently, a nearest-neighborhood matching algorithm was applied to pair each treated subject with a control subject based on the logit of the estimated propensity score. This matching was conducted without replacement, employing a caliper width equal to $0.2$ times the standard deviation of the logit propensity score \citep{austin_optimal_2011}. For CEM, we utilized an auto-coarsening strategy based on the Struges' rule (``Auto'') and a less coarsening rule of dividing each matching variable into 3 equal groups (``G3''). The implementation of PSM and CEM was carried out in $\texttt{R/MatchIt}$. For each matched dataset, we conducted the analyses described below.

\textbf{\textit{Imbalance, Bias, and Sample size}:} In this analysis, we estimate a homogeneous treatment effect using unadjusted analyses under both CEM and PSM designs. The treatment and outcome models included only additive linear terms of five variables (Scenario 1, Web Appendix 1), with the treatment effect size $\beta_1$ fixed at 6 in the outcome model (\ref{eq_1}). For each matched dataset, we computed the SMD-based multivariate imbalance metrics using formula (\ref{eq_4}) and recorded the matched sample sizes. Treatment effects were then estimated by fitting $\mathcal{M}(W)$ under each design. By evaluating the performance of $\mathcal{M}(W)$ in each design, we assess whether matching alone can minimize imbalance and residual confounding bias.

To assess whether imbalance under CEM is sensitive to sample size, we fixed one pair of $(\bm{\alpha}_1, \bm{\beta}_2)$ from the 300 possible combinations to generate unmatched cohorts with sample sizes ranging from 2,500 to 12,500 and computed the Mahalanobis and SMD distances under CEM at each size.

\textbf{\textit{Model dependence}:} To compare the robustness of CEM (``auto'') and PSM to model misspecification, we introduced nonlinearity into both the outcome model (\ref{eq_1}) and the treatment model (\ref{eq_2}) in Scenario 2 (Web Appendix 1). In this scenario, both models include quadratic and interaction terms for five variables, and the treatment effect $\beta_1$ is fixed at 6, making $\mathcal{M}(W,\bm{X})$ a misspecified model. Both scenarios assume homogeneous treatment effects. In Scenario 1, we fit the correctly specified model $\mathcal{M}(W,\bm{X})$ under CEM (``auto'') and PSM, and computed the resulting biases to compare performance when the model is correctly specified. As a baseline, we also calculated the biases of $\mathcal{M}(W)$ and $\mathcal{M}(W,\bm{X})$ in the unmatched cohorts to show that both CEM and PSM can reduce dependence on correct model specification for unbiased estimation. In Scenario 2, we evaluated the performance of the misspecified model $\mathcal{M}(W,\bm{X})$ under both CEM and PSM. By examining this performance across designs, we assess which design is more robust to model misspecification. This comparison highlights that CEM (``auto'') may yield larger bias when the outcome model is misspecified, due to incorrect adjustment for matched confounders, whereas PSM appears more robust to such bias.

\textbf{\textit{Curse of dimensionality and trade-off with coarsening}:} To examine the sparsity bias CEM may face due to the curse of dimensionality, we increased the number of matching variables to seven and generated data using the treatment and outcome models specified in Scenarios 3 and 4 (Web Appendix 1). In Scenario 3, both models include only additive linear terms of the seven variables, and $\mathcal{M}(W,\bm{X})$ is correctly specified. We applied CEM(``auto''), CEM(``G3''), and PSM. Comparing CEM(``auto'') to PSM allows us to show that CEM(``auto'') suffers from sparsity bias in high-dimensional settings, even with a correctly specified outcome model, whereas PSM avoids this issue. The comparison between CEM(``G3'') and PSM illustrates whether reduced coarsening can help mitigate sparsity bias. In Scenario 4, both models include quadratic and interaction terms of the seven variables, rendering $\mathcal{M}(W,\bm{X})$ misspecified. Here, we compare CEM(``G3'') and PSM to demonstrate that while greater coarsening may help CEM alleviate sparsity issues, it does so at the cost of increased residual confounding and model dependence. As a result, CEM(``G3'') is expected to perform worse than PSM. For each matched dataset, we fit $\mathcal{M}(W,\bm{X})$ and computed the bias, variance, and MSE of the estimated treatment effect across methods.

\textbf{\textit{Heterogeneous treatment effect}:} We conducted a simulation study to evaluate the bias in estimating the PATT when treatment effect heterogeneity is present. Two variables, $X_1$ and $X_2$, were generated from a standard normal distribution. Data were generated under both linear and nonlinear treatment assignment and outcome models, with two treatment prevalence scenarios ($\sim 10\%$ and $\sim 30\%$) determined by varying the intercept in the propensity model. The treatment effect depended on $X_1$, creating heterogeneity such that individuals with higher $X_1$ values experienced larger effects. True PATT values were obtained by simulating $10^8$ observations and averaging individual treatment effects among treated units. (Details in Web Appendix 2)

For each simulated dataset, three matching strategies—CEM (``auto''), CEM (``G3''), and PSM with a caliper of 0.2 $\times$ the logit propensity score were applied. PATT was then estimated using either a unadjusted model $\mathcal{M}(W)$ or the following formula:
\begin{align}
	\hat{\tau}_{PATT}=\hat{\beta}_1+\hat{\theta}\cdot \hat{\mathbb{E}}(X_i|W_i=1)
	\label{eq_5}
\end{align}
where $\hat{\beta}_1$ and $\hat{\theta}$ were estimated using the interaction model $\mathcal{M}(W, W \cdot X_1, X_1, X_2)$ (including $W$, $W$ by $X_1$ interaction, linear terms of $X_1$ and $X_2$), which is misspecified under nonlinear outcome generation. $\hat{\mathbb{E}}(X_i|W_i=1)$ is the estimated conditional mean of $X_1$ among the treated. When the treated proportion is $\sim 10\%$, one treated subject could have nine untreated subjects available for pairing, and thus the likelihood that all treated subjects enter the matched data is high, allowing $\mathcal{M}(W)$ to estimate $\hat{\tau}_{PATT}$ with minimal bias in matching designs because of small bias in estimating $\mathbb{E}(X_i|W_i=1)$ using sample mean of $X_1$ among matched treated subjects. However, when the treated proportion is high (i.e.,$\sim 30\%$), the chance that some treated subjects are discarded because they fail to pair with untreated subjects increases. In this case, the distribution of $X_1$ among the treated subjects in the matched data may not represent the distribution of $X_1$ among the treated subjects in the original population, and the sample mean of $X_1$ among the treated subjects in the matched data will be a biased estimate of $\mathbb{E}(X_i|W_i=1)$. In this situation, $\mathcal{M}(W)$ will be biased. To correct this bias, we substitute formula (\ref{eq_5}) with the sample mean of $X_1$ among the treated subjects in the unmatched data as follows:
\begin{align}
	\hat{\tau}_{PATT}=\hat{\beta}_1 + \hat{\theta} \cdot \frac{\sum_{i=1}^{N} X_{i1} \cdot I(W_i = 1)}{\sum_{i=1}^N I(W_i = 1)}
	\label{eq_6}
\end{align}
where $N$ is the sample size of the unmatched cohort and $I(\cdot)$ is the indicator function. This process was repeated 2,000 times, and estimates were averaged to assess performance across designs and scenarios.

A key distinction between our simulation and the earlier study by Iacus et al. \citep{iacus_causal_2012} is that, unlike their work—which used a misspecified model to estimate the propensity score—we employed a correctly specified model. In their case, the estimated propensity scores do not satisfy the balancing score or SITA assumptions. Our study seeks to evaluate whether PSM—despite possessing the balancing score and SITA properties-is, as suggested in the clinical literature \citep{haider_comparative_2013, Muennig_2016, wharam_breast_2018, sceats_nonoperative_2019, liu_propensity-score_2023, george_outcomes_2024}, theoretically inferior to CEM in terms of balancing matched confounders and reducing model dependence and bias.

\subsection{Simulation results}
\label{sec32}

Figure \ref{fig1} presents the comparison between CEM and PSM using the unadjusted model $\mathcal{M}(W)$.  Unadjusted estimates are less biased in PSM than in CEM, and CEM(``auto'') yields less biased estimates than CEM(``G3'') (Figure \ref{fig1}a). This is because CEM tends to yield greater covariate imbalance than PSM, and this imbalance contributes to residual confounding (Figure \ref{fig1}b). Reducing the level of coarsening, from auto-coarsening to grouping into three categories, increases sample size but also increases imbalance (Figures \ref{fig1}b and \ref{fig1}c). As shown in Figure \ref{fig1}d, the imbalance in CEM, measured using both SMD-based and Mahalanobis distance metrics, does not converge to zero as the sample size increases. Unlike in PSM, this imbalance in CEM is systematic rather than random and an increase in imbalance tends to amplify residual confounding bias.

Figure \ref{fig2} compares CEM (``auto'') and PSM using the adjusted model $\mathcal{M}(W,\bm{X})$. When $\mathcal{M}(W,\bm{X})$ represents the true outcome model (Figure \ref{fig2}a), both PSM and CEM(``auto'') yield small biases, with PSM performing slightly better—possibly due to its larger matched sample size. However, when $\mathcal{M}(W,\bm{X})$ is misspecified (Figure \ref{fig2}b), the performance of CEM(``auto'') deteriorates relative to PSM. Still, compared to the unmatched design (also in Figure \ref{fig2}c), both CEM(``auto'') and PSM substantially reduce bias, supporting the notion that a good matching design can mitigate model dependence, even when matching is not exact.

Figure \ref{fig3} compares CEM (``auto'' and ``G3'') and PSM using the adjusted model $\mathcal{M}(W,\bm{X})$ in settings where the number of matched confounders is relatively large given the sample size. With auto-coarsening, CEM suffers from severe data loss due to the curse of dimensionality. As a result, even when $\mathcal{M}(W,\bm{X})$ is correctly specified, effect estimates can become highly volatile and biased, often yielding extreme values. Variance and MSE estimates of $\mathcal{M}(W,\bm{X})$ are similarly unstable (Figures \ref{fig3}A–C). Increasing the level of coarsening—for example, by grouping confounders into three categories—can mitigate data loss. In this case, and assuming a correctly specified outcome model, CEM can slightly outperform PSM in terms of bias, precision, and MSE (Figures \ref{fig3}D–F), owing to its larger matched sample size. However, this increased coarsening also raises imbalance and amplifies residual confounding, making correct model specification essential to mitigate bias. As shown in Figures \ref{fig3}G–I, when $\mathcal{M}(W,\bm{X})$ is misspecified, CEM performs worse than PSM, exhibiting greater bias, variance, and MSE.

Table \ref{table1} summarizes results across four simulation scenarios varying the outcome model specification and proportion treated. When $\mathcal{M}(W, W \cdot X_1, X_1, X_2)$ represents the true outcome model (Scenarios 1–2), model-based estimates using $\mathcal{M}(W)$ in PSM consistently yield lower bias, variance, and RMSE than in CEM (``auto'') and especially CEM (``G3''), where coarser grouping produces larger residual confounding. Analytic formula–based estimate using formula (\ref{eq_6}) perform similarly across all three designs because $\mathcal{M}(W, W \cdot X_1, X_1, X_2)$ produces unbiased estimates of $\beta_1$ and $\theta$. When proportion treated is low, analytic estimates in PSM and CEM (``auto'') are generally comparable to the corresponding model-based estimates. When proportion treated is high, analytic estimates in PSM and CEM (``auto'') outperforms the corresponding model-based estimates. When the outcome model is misspecified (Scenarios 3–4), analytic estimates in PSM outperforms CEM (``auto'') and especially CEM (``G3'') because misspecified model $\mathcal{M}(W, W \cdot X_1, X_1, X_2)$ in CEM give more biased estimates of $\beta_1$ and $\theta$ than in PSM. In these misspecified settings, analytic estimates outperform the corresponding model-based estimates in PSM and CEM (``auto''). Both estimates have large biases in CEM(``G3'') due to coarsening and residual confounding.  Overall, analytic formula (\ref{eq_6}) in PSM performs consistently the best across four scenarios.

In summary, we observe similar patterns for both homogeneous and heterogeneous treatment effects when using PSM or CEM. When matching involves few variables, PSM performs comparably to—or slightly better than—CEM with auto-coarsening, and both are similarly robust to model misspecification when residual confounding is small. CEM (``G3'') performs substantially worse due to greater residual confounding and heightened model dependence. Under high dimensionality, fine coarsening (e.g., auto-coarsening) is infeasible. Additional coarsening reduces data loss but introduces substantial residual confounding, increases model dependence, and raises the risk of bias from model misspecification in CEM.

\section{Discussion}
\label{sec6}

We have shown that the purported advantages of CEM over PSM do not consistently hold, particularly as the number of matched confounders increases. CEM, an inexact covariate-matching method due to its coarsening step, exhibits poorer covariate balance than PSM, even with auto-coarsening. This imbalance is structural, persisting even as sample size grows, and leads to residual confounding, a key limitation absent from the discussions in previous studies \citep{iacus_causal_2012,Ripollone_2019,Black_2020}. The curse of dimensionality further exacerbates the problem: while auto-coarsening may yield comparable performance to PSM with few covariates, its performance deteriorates as dimensionality increases, causing data loss, unstable effect estimates, and greater model dependence. Although more aggressive coarsening can mitigate data loss, it amplifies residual confounding and statistical bias. In contrast, PSM achieves better covariate balance, lower bias, and reduced reliance on outcome modeling. By framing the comparison through the theoretical properties of the two designs rather than relying solely on simulations, this study clarifies why CEM’s structural imbalance makes it different from, and often less effective than PSM, despite common misconceptions that it is an exact matching method.

Frequency matching, a coarsened matching variant often used in case-control studies, is similarly prone to residual confounding due to its inherent inexactness \citep{wan_mat_2021,wan_conditional_2022,wan_does_2024}. The CEM literature often overlooks this issue, and CEM is frequently misconceived as an exact matching design. For instance, Sceats et al. claimed that CEM mimics a fully blocked randomized experiment \citep{sceats_nonoperative_2019}, and Blackwell et al. \cite{blackwell_cem:_2009} suggested that a simple mean difference suffices to estimate treatment effects. Iacus et al. \citep{iacus_causal_2012} argued that CEM rests on the principle of grouping ``substantively indistinguishable values.'' However, this relies on a problematic assumption: that matching on $\bm{C}$ is equivalent to matching on $\bm{X}$. This leads to a form of circular logic—coarsening reduces dimensionality, but to justify its validity, one must assume that coarsening preserves all relevant information in $\bm{X}$. In high-dimensional or continuous settings, this assumption is unrealistic and undermines the rationale for coarsening. In practice, coarsening choices are driven more by feasibility than by theoretical justification, making residual confounding difficult to avoid. The magnitude of such bias in CEM ultimately depends on the degree of coarsening.

Our study also addresses the common concern that PSM increases imbalance and that CEM offers superior covariate balance \citep{king_why_2019,Ripollone_2019}. These claims often rely on inappropriate metrics, such as the Mahalanobis distance, which fails to account for the direction of chance imbalance. Using a multivariate SMD metric, we show that PSM with an appropriately chosen caliper achieves better covariate balance than CEM with auto-coarsening. Moreover, unadjusted analysis is suitable for PSM, which inherently balances confounders \citep{rosenbaum_central_1983}, but not for CEM, where coarsening introduces systematic imbalance. Crucially, this imbalance in CEM does not diminish with increasing sample size. In contrast, imbalance in exactly matched PSM designs arises from random variation and converges to zero as sample size increases \cite{rosenbaum_design_2020,wan_psmparadox_2025}. Mitigating residual confounding in CEM requires accurate outcome modeling, which increases model dependence. When outcome models are misspecified, CEM yields more biased estimates than PSM. By contrast, PSM can produce unbiased estimates even under a class of misspecified models.

This highlights a broader point: a good matching design can reduce reliance on correct model specification, even when matching is inexact \citep{ho_matching_2007,guo_statistical_2023}. Our results support this view, demonstrating that both PSM and covariate matching reduce model dependence relative to unmatched designs. However, exact covariate matching becomes impractical in high dimensions due to a rapidly shrinking number of matches \citep{abadie_matching_2016}. Similarly, CEM with auto-coarsening suffers from data loss, increased bias, and reduced precision as dimensionality increases \citep{Black_2020, Ripollone_2019}, a pattern corroborated in our simulations. As the number of matching variables grows, CEM becomes susceptible to sparsity bias and requires increasingly aggressive coarsening. Although greater coarsening can mitigate sparsity bias, it also amplifies residual confounding and can introduce substantial bias if the outcome–confounder relationship is not correctly modeled in the post-matching analysis.

We conclude that PSM is more robust than CEM when the balancing score and SITA properties hold. However, like any statistical design or method, PSM has limitations, and alternatives such as weighting could be explored \citep{Busso_2014}. For example, PSM cannot address unmeasured confounding, making sensitivity analyses necessary to assess its influence. Prior research has shown that PSM can yield more biased estimates than CEM when the propensity score model is misspecified \citep{iacus_causal_2012,lenis_measuring_2018}. A misspecified propensity score may violate the balancing and SITA assumptions, leaving residual confounding after matching. To mitigate this risk, flexible machine learning methods—such as generalized boosted regression—can be used to estimate the propensity score, thereby capturing nonlinearities and interactions among covariates \citep{McCaffrey_2004}.


\newpage

\section{Appendix}

\begin{customprop}{1}
	\label{prop1}
	Suppose we fit a linear regression model in an exactly matched PSM design to estimate ATT, defined as follows:
	\begin{align*}
		\mathbb{E}(Y_i|\boldsymbol{\tilde{X}_i})=\gamma_0+\gamma_1 W_i +  \boldsymbol{\tilde{X}_i} \boldsymbol{\gamma_{2}}^\top\, , \,  i=1,2,\cdots, n,
	\end{align*}
	where $\boldsymbol{\tilde{X}}_i$ is a $1 \times K$ subset of $\boldsymbol{X}$, with $0 \leq K \leq p$. $\boldsymbol{\tilde{X}}_i$ may be an empty set or include any combination of confounders $\boldsymbol{X}$. $\boldsymbol{\gamma_{2}}$ is $1 \times K$ vector of regression coefficients. Here, $n$ represents the sample size of the PSM design. The resulting OLS estimator satisfies
	\begin{align*}
		\hat{\gamma}_1 \overset{p}{\to} \mathbb{E}(Y(1)-Y(0)|W=1) ,
	\end{align*} 
\end{customprop}

\begin{proof}
	Let $\boldsymbol{1}$ denote a $n\times1$ vector of 1s,  $\boldsymbol{\tilde{X}}$ denote $n \times K$ matrix of adjusted confounders, $\boldsymbol{W}$ denote $n\times1$ vector of the treatment status, $\boldsymbol{Y}$ denote $N\times1$ vector of the outcome.	The $N\times(K+2)$ design matrix for the linear model is $[\boldsymbol{1} \, \boldsymbol{\tilde{X}} \, \boldsymbol{W}]$.  We let $S$ denote the status of being selected into the PSM design from the source data. We further denote the first two components $[1 \, \boldsymbol{\tilde{X}}]_{n\times(K+1)}$ as $Z_1$, and the last column as $Z_2$. Thus, the design matrix can be rewritten as $[Z_1 \, Z_2]$, and the least-square estimator of 
	$\gamma_1$  can be derived as:
	\begin{align*}
		\hat{\gamma}_{1}=(Z^T_{2|1}Z_{2|1})^{-1}Z^{T}_{2|1}\boldsymbol{Y}
	\end{align*}
	where $Z_{2|1}=(I-Z_1(Z^T_1Z_1)^{-1}Z^T_1)Z_2$. $H=Z_1(Z^T_1Z_1)^{-1}Z^T_1$ is the hat matrix. $HZ_2$ is the projection of $W$ onto the space spanned by $\boldsymbol{\tilde{X}}$. $Z_{2|1}$ can be interpreted as the residual from fitting a linear regression model of confounders $\boldsymbol{\tilde{X}}$ on $W$ in the matched data. we proceed as follows:
	
	\begin{itemize}
		\item[(1)] {\it Derivation of probability limit of projection $HZ_2$.}  We relabel $K$ confounders in $\boldsymbol{\tilde{X}}$ as $X_1, X_2, \cdots, X_K$. The closed-form matrix inversion becomes highly complex with more than one $X$ variables. For simplicity, we use mean-centered variables, $\boldsymbol{\tilde{X}^*}=[X^*_1, X^*_2, \cdots, X^*_K ]$ , in the PSM design. A demonstration for a single $X$ without mean-centering is provided in the web material. The design matrix \( Z_1 \) is:
		\[
		Z_1 = \begin{bmatrix} 
			1 & X^*_{1,1} & X^*_{2,1} &\hdots & X^*_{K,1} \\
			1 & X^*_{1,2} & X^*_{2,2} &\hdots & X^*_{K,2} \\
			\vdots & \vdots & \vdots &\hdots & \vdots \\
			1 & X^*_{1,N} & X^*_{2,N} &\hdots & X^*_{K,N}
		\end{bmatrix}.
		\]
		
		\[
		Z^T_1 Z_1 =
		\begin{bmatrix}
			N & \sum X^*_{1,i} & \sum X^*_{2,i} & \hdots & \sum X^*_{K,i} \\
			\sum X^*_{1,i} & \sum X^{*2}_{1,i} & \sum X^*_{1,i} X^*_{2,i} &\hdots & \sum X^*_{1,i} X^*_{K,i} \\
			\sum X^*_{2,i} & \sum X^*_{1,i} X^*_{2,i} & \sum X^{*2}_{2,i} &\hdots & \sum X^*_{2,i} X^*_{K,i} \\
			\vdots & \vdots & \vdots &\hdots & \vdots \\
			\sum X^*_{K,i} & \sum X^*_{1,i} X^*_{K,i} & \sum X^*_{2,i}X^*_{K,i} &\hdots & \sum X^{*2}_{K,i}
		\end{bmatrix}.
		\]
		This matrix captures the sums of squares and cross-products of $Z_1$. The inverse $(Z^T_1 Z_1)^{-1}$ exists if $Z$ has full rank and yields a $(K+1) \times (K+1)$ matrix. 
		
		Next, we compute \( Z^T_1 Z_2 \) as
		\[
		Z^T_1 Z_2 = 
		\begin{bmatrix}
			\sum_{i=1}^n W_i \\
			\sum_{i=1}^n W_i X^*_{1,i} \\
			\sum_{i=1}^n W_i X^*_{2,i} \\
			\vdots                   \\
			\sum_{i=1}^n W_i X^*_{K,i}
		\end{bmatrix}.
		\]
		
		$\frac{1}{N} (Z^T_1 Z_1) \overset{p}{\to} \mathbb{E}(Z^T_1 Z_1)$, where 
		\[
		\mathbb{E}(Z^T_1 Z_1) =
		\begin{bmatrix}
			1 & \mathbb{E}(X^*_{1,i}|S_i=1) & \mathbb{E}(X^*_{2,i}|S_i=1) & \hdots & \mathbb{E}(X^*_{K,i}|S_i=1) \\
			\mathbb{E}(X^*_{1,i}|S_i=1) & \mathbb{E}(X^{*2}_{1,i}|S_i=1) & \mathbb{E}(X^*_{1,i} X^*_{2,i}|S_i=1) &\hdots & \mathbb{E}(X^*_{1,i} X^*_{K,i}|S_i=1) \\
			\mathbb{E}(X^*_{2,i}|S_i=1) & \mathbb{E}(X^*_{1,i} X^*_{2,i}|S_i=1) & \mathbb{E}(X^{*2}_{2,i}|S_i=1) &\hdots & \mathbb{E}(X^*_{2,i} X^*_{K,i}|S_i=1) \\
			\vdots & \vdots & \vdots &\hdots & \vdots \\
			\mathbb{E}(X^*_{K,i}|S_i=1) & \mathbb{E}(X^*_{1,i} X^*_{K,i}|S_i=1) & \mathbb{E}(X^*_{2,i} X^*_{K,i}|S_i=1)&\hdots & \mathbb{E}(X^{*2}_{K,i}|S_i=1)
		\end{bmatrix}.
		\]
		Here, $S_i=1$ indicates selection into the matched design. Since $X^*_{k,i}$ is defined as mean-centered, we have $\mathbb{E}(X^*_{k,i}|S_i=1)=0, \, \forall \, k=1,2,\cdots, K$. Next,
		\[
		\frac{1}{N}Z_1^T Z_2  = 
		\begin{bmatrix}
			\frac{\sum_{i=1}^n W_i}{n} \\
			\frac{\sum_{i=1}^n W_i X^*_{1,i}}{n} \\
			\frac{\sum_{i=1}^n W_i X^*_{2,i}}{n} \\
			\vdots                  \\
			\frac{\sum_{i=1}^n W_i X^*_{K,i}}{n}
		\end{bmatrix}.
		\]
		
		$\frac{1}{n}Z_1^T Z_2 \overset{p}{\to} \mathbb{E}(Z^T Z_2)$, where 
		\[
		\mathbb{E}(Z_1^T Z_2) = 
		\begin{bmatrix}
			\mathbb{E}(W_i|S_i=1) \\
			\mathbb{E}(W_i X^*_{1,i}|S_i=1) \\
			\mathbb{E}(W_i X^*_{2,i}|S_i=1) \\
			\vdots                   \\W
			\mathbb{E}(W_i X^*_{K,i}|S_i=1)
		\end{bmatrix}.
		\]
		
		Under the balancing score property of propensity score $\boldsymbol{X} \independent W | e(\boldsymbol{X}) $, it follows that each $X$ in $\boldsymbol{\tilde{X}}$ is not associated with $W$ in the PSM design. subsequently, $\mathbb{E}(W_i X^*_{k,i}|S_i=1) =0, \, \forall \, k=1,2,\cdots, K$. 	$\mathbb{E}(W_i|S_i=1)=\delta$, the proportion of the treated subjects in the matched design. e.g., for a 1:1 matching design, $\delta=0.5$. It follows that
		
		\[
		\mathbb{E}(Z_1^T Z_2) = 
		\begin{bmatrix}
			\delta \\
			0\\
			0\\
			\vdots                   \\
			0
		\end{bmatrix}.
		\]
		and
		\[
		(Z^T_1 Z_1)^{-1} Z_1^T Z_2 \overset{p}{\to} \mathbb{E}(Z^T_1 Z_1)^{-1} \mathbb{E}(Z_1^T Z_2) =
		\begin{bmatrix}
			\delta \\
			0\\
			0\\
			\vdots                   \\
			0
		\end{bmatrix}.
		\]
		Here, $(Z_1^T Z_1)^{-1} Z_1^T Z_2$ represents the estimated regression coefficients for the model fitting $W$ on $\boldsymbol{\tilde{X}}$. In the PSM-matched design, the coefficients of $\boldsymbol{\tilde{X}}$ are asymptotically zero due to balance or association with $W$. The first element corresponds to the intercept, representing the grand average of $W$, or the treatment proportion in the matched design.
		
		\[HZ_2= Z_1 \big( (Z^T_1 Z_1)^{-1}Z_1^T Z_2 \big) \overset{p}{\to} Z_1 
		\begin{bmatrix}
			\delta \\
			0\\
			0\\
			\vdots                   \\
			0
		\end{bmatrix}   \\
		=\delta \boldsymbol{1}
		\]
		
		\item[(2)] {\it Derivation of probability limit of $Z_{2|1}$}
		
		\[Z_{2|1}=Z_2-\delta \boldsymbol{1}=
		\begin{bmatrix}
			W_1-\delta \\
			W_2-\delta \\
			W_3-\delta \\
			\vdots                   \\
			W_K-\delta 
		\end{bmatrix}.
		\]
		
		$Z_{2|1}^TZ_{2|1} =\sum^n_{i=1}W^2_i-2\delta \sum^n_{i=1}W_i+\sum^n_{i=1}\delta^2=\sum^n_{i=1}W_i-2\delta \sum^n_{i=1}W_i+N\delta^2 $ and  $Z^T_{2|1}Y=\sum^n_{i=1}(W_i-\delta)Y_i =\sum^n_{i=1}(W_i Y_i)-\delta \sum^n_{i=1}Y_i$

		\item[(3)] {\it Derivation of probability limit of projection $(Z_{2|1}^TZ_{2|1})^{-1}Z^T_{2|1}Y$.} 
		
		It follows that
		\begin{align*}
			(Z_{2|1}^TZ_{2|1})^{-1}Z^T_{2|1}Y &=\frac{(\sum^n_{i=1}(W_i Y_i)-\delta \sum^n_{i=1}Y_i)/n}{(\sum^n_{i=1}W_i-2\delta \sum^n_{i=1}W_i+n\delta^2)/n } \\
			&\overset{p}{\to} \frac{\mathbb{E}(W_iY_i|S_i=1)-\mathbb{E}(W_i|S_i=1)\mathbb{E}(Y_i|S_i=1)}{\delta(1-\delta)} \\
			&=\frac{Cov(W_i,Y_i|S_i=1)}{Var(W_i|S_i=1)}
		\end{align*}
		$\frac{Cov(W_i,Y_i|S_i=1)}{Var(W_i|S_i=1)} $ represents the coefficient of $W$ from an univariate regression of $Y$ on $W$ in the PSM design. Furthermore, 
		\begin{align*}
			\mathbb{E}(W_iY_i|S_i=1) &=Y_i \mathbb{P}(Y_i,W_i=1|S_i=1) =Y_i \mathbb{P}(Y_i|W_i=1,S_i=1)\mathbb{P}(W_i=1|S_i=1)\\
			&=\mathbb{E}(Y_i|W_i=1,S_i=1)\mathbb{P}(W_i=1|S_i=1) =\mathbb{E}(Y_i|W_i=1,S_i=1)\mathbb{E}(W_i|S_i=1)
		\end{align*}
		and 
		\begin{align*}
			\mathbb{E}(Y_i|S_i=1) &=\mathbb{E}(Y_i|W_i=1,S_i=1)\mathbb{P}(W_i=1|S_i=1) + \mathbb{E}(Y_i|W_i=0,S_i=1)\mathbb{P}(W_i=0|S_i=1) \\
			&=\mathbb{E}(Y_i|W_i=1,S_i=1)\mathbb{E}(W_i|S_i=1) + \mathbb{E}(Y_i|W_i=0,S_i=1)(1-\mathbb{E}(W_i|S_i=1))
		\end{align*}
		Thus, 
		\begin{align*}
			\mathbb{E}(W_iY_i|S_i=1)-\mathbb{E}(W_i|S_i=1)\mathbb{E}(Y_i|S_i=1) &= (\mathbb{E}(Y_i|W_i=1,S_i=1)-\mathbb{E}(Y_i|W_i=0,S_i=1)) \mathbb{E}(W_i|S_i=1)(1-\mathbb{E}(W_i|S_i=1)) \\
			&=(\mathbb{E}(Y_i|W_i=1,S_i=1)-\mathbb{E}(Y_i|W_i=0,S_i=1)) \delta(1-\delta) 
		\end{align*}
		
		It follows that
		\begin{align*}
			(Z_{2|1}^TZ_{2|1})^{-1}Z^T_{2|1}Y &\overset{p}{\to} \mathbb{E}(Y_i|W_i=1,S_i=1)-\mathbb{E}(Y_i|W_i=0,S_i=1) \\
			&=\mathbb{E}(Y_i(1)|W_i=1,S_i=1)-\mathbb{E}(Y_i(0)|W_i=0,S_i=1) \\
			&=\mathbb{E}(Y_i(1)|W_i=1,S_i=1)-\mathbb{E}(Y_i(0)|W_i=1,S_i=1) \\
			&= \mathbb{E}(Y_i(1)-Y_i(0)|W_i=1,S_i=1) \\
			&= \mathbb{E}(Y_i(1)-Y_i(0)|W_i=1)
		\end{align*}
		Last equality holds because under the strongly ignorability condition $\{ Y(1),Y(0) \} \independent W | e(\bm{X})$, $\mathbb{E}(Y_i(0)|W_i=0,S_i=1)=\mathbb{E}(Y_i(0)|W_i=1, S_i=1)$. The chance of being selected into the matched design ($S_i=1$) for the treated does not depend on either $Y_i(1)$ or $Y_i(0)$. e.g., we normally select all treated subjects into the matched design when estimating ATT.

	\end{itemize}

\end{proof}

\newpage

\begin{figure}[h!]
	\centering
	\includegraphics[width=\linewidth]{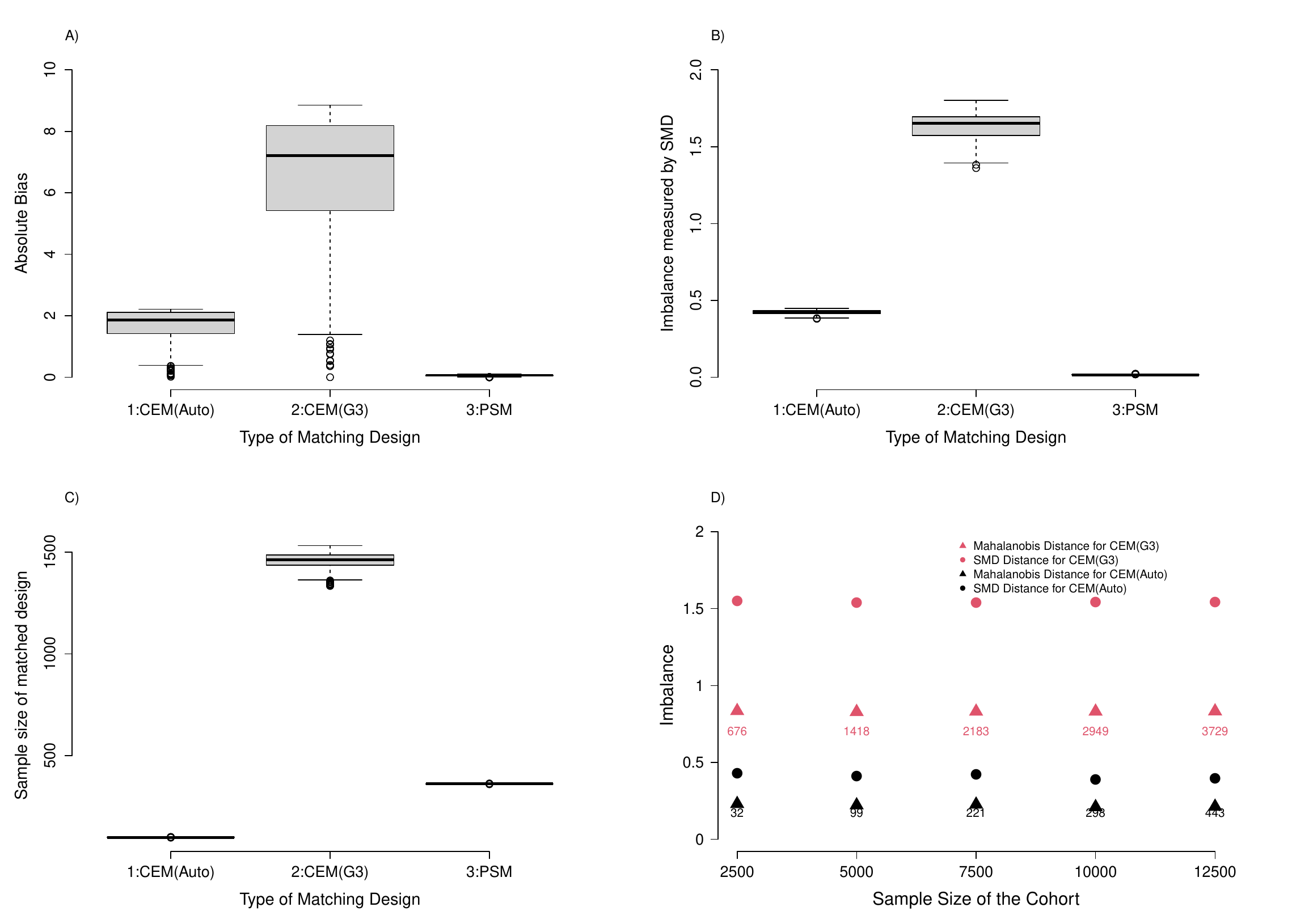}
	\caption{A) Absolute biases of unadjusted estimates in CEM (``Auto''), CEM (``G3''), and PSM with a caliper width equal to $0.2$ times the standard deviation of the logit propensity score; B) SMD imbalance measures for CEM (``Auto''), CEM (``G3''), and PSM; C) Sample sizes for CEM (``Auto''), CEM (``G3''), and PSM; D) Mahalanobis and SMD imbalance measures for CEM (``Auto'') and CEM (``G3''). The numbers beneath the red and black triangle symbols represent the average sample sizes of the matched cohorts.}
	\label{fig1}
\end{figure}

\begin{figure}[h!]
	\centering
	\includegraphics[width=\linewidth]{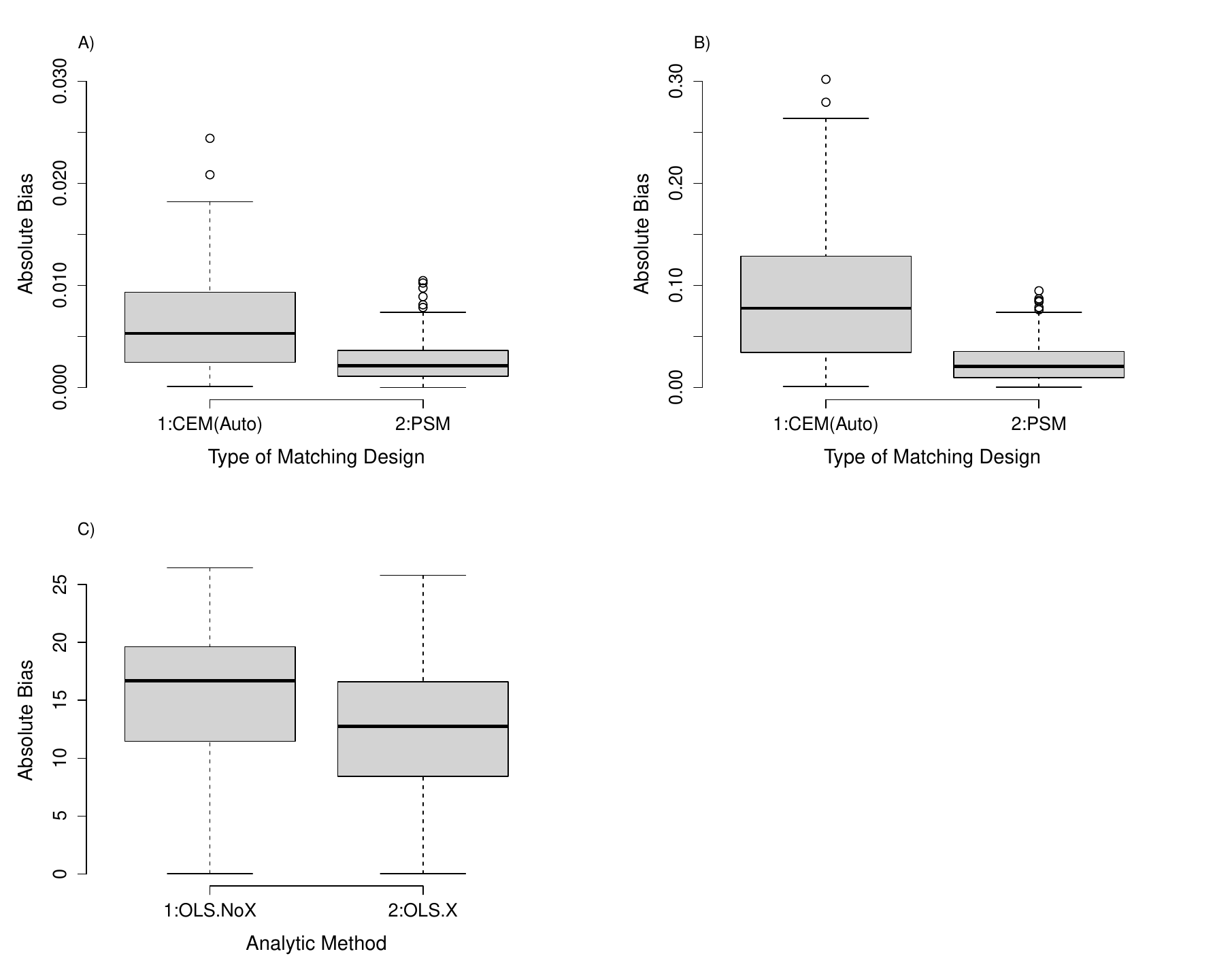}
	\caption{A) Absolute biases of adjusted estimates in CEM (``Auto'') and PSM when $\mathcal{M}(W,\bm{X})$ is the correct model ; B) Absolute biases of adjusted estimates in CEM (``Auto'') and PSM when $\mathcal{M}(W,\bm{X})$ is the misspecified model; C) Absolute biases of $\mathcal{M}(W)$ and $\mathcal{M}(W,\bm{X})$ in unmatched cohorts.}
	\label{fig2}
\end{figure}

\begin{figure}[h!]
	\centering
	\includegraphics[width=\linewidth]{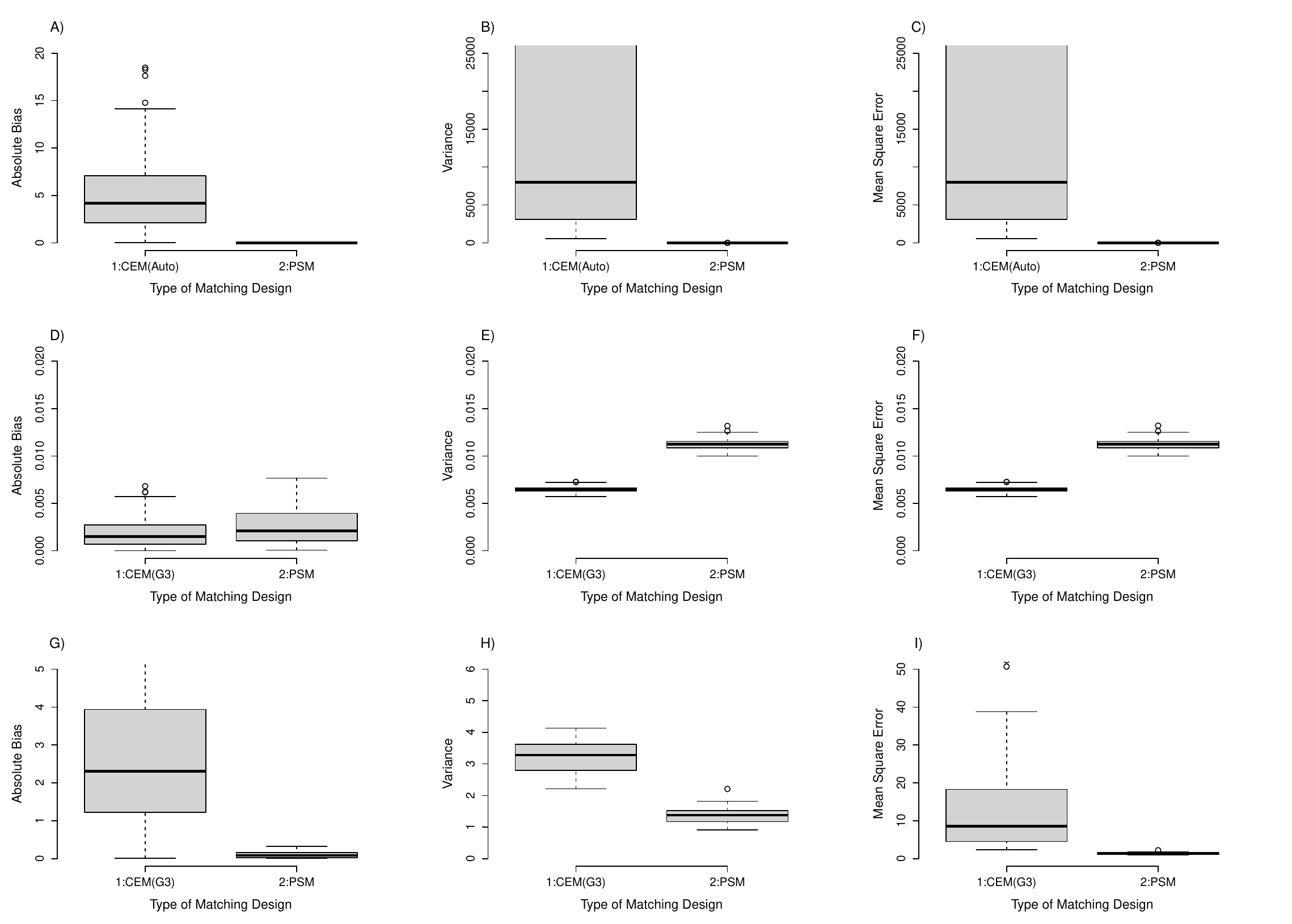}
	\caption{A–C: Absolute bias, variance, and MSE of adjusted estimates in CEM(``Auto'') and PSM when $\mathcal{M}(W,\bm{X})$ is the true model. D–F: Same metrics for CEM(``G3'') and PSM when $\mathcal{M}(W,\bm{X})$ is the true model. G–I: Same metrics for  CEM(``G3'') and PSM when $\mathcal{M}(W,\bm{X})$ is the misspecified model.}
	\label{fig3}
\end{figure}

\begin{table}[h!]
	\begin{center} 
		\caption{Simulation results for the estimation of heterogeneous treatment effects} 
		\begin{tabular}{c c c c c c c c} 
			\toprule
			Scenario No&Proportion of treated &Design &Model & Mean Estimate & Bias & SD & Root MSE\\ [0.5ex] 
			\hline 
			$1^*$ &$\sim 10\%$ & PSM & $\mathcal{M}(W)$ & 0.7353 &0.0136 &0.0764 &0.0776 \\ 
			& &     & Formula (\ref{eq_6}) &0.7197 &-0.0020 &0.0695 &0.0695\\
			& & CEM(``Auto'') &$\mathcal{M}(W)$ &0.7579 &0.0363 &0.0699 &0.0787 \\
			& &     & Formula (\ref{eq_6}) &0.7221 &0.0004 &0.0837 &0.0702\\		
			& & CEM(``G3'') &$\mathcal{M}(W)$ &1.2322&0.5105&0.0837&0.5174 \\
			& &     & Formula (\ref{eq_6}) &0.7198 &-0.0019 &0.0780 &0.0780
			
			\\ [1ex] 
			2& $\sim 30\%$ & PSM &$\mathcal{M}(W)$ &0.9662 &0.1079 &0.0441 &0.1166 \\
			& &     & Formula (\ref{eq_6})&0.8595 &0.0012 &0.0423 &0.0423 \\
			& & CEM(``Auto'') &$\mathcal{M}(W)$ &0.9744 &0.1162 &0.0421 &0.1236 \\
			& &     & Formula (\ref{eq_6})  &0.8588 &0.00051 &0.0476 &0.0432 \\
			& & CEM(``G3'') &$\mathcal{M}(W)$ &1.3300 &0.4717 &0.0476 &0.4741\\
			&  &   & Formula (\ref{eq_6}) &0.8590 &0.0007 &0.0455 &0.0455  \\                	
			\\	
			
			3&$\sim 10\%$ & PSM & $\mathcal{M}(W)$ &1.6023 &-0.1487 &0.0782 &0.1680\\ 
			& &     & Formula (\ref{eq_6}) &1.7458 &-0.0052 &0.0793 &0.0794\\
			& & CEM(``Auto'') &$\mathcal{M}(W)$ &1.5746 &-0.1764 &0.0755 &0.1918     	 \\
			& &     & Formula (\ref{eq_6})  &1.7109 &-0.0401 &0.1027 &0.0871\\
			& & CEM(``G3'') &$\mathcal{M}(W)$ &0.4362 &-1.3147 &0.1027 &1.3187\\
			&  &   & Formula (\ref{eq_6}) &0.6419 &-1.1091 &0.1292 &1.1166 \\                	
			
			\\ 	
			
			4&$\sim 30\%$ & PSM & $\mathcal{M}(W)$ &1.2808 &-0.2028 &0.0459 &0.2079 \\ 
			& &     & Formula (\ref{eq_6}) &1.4785 &-0.0052 &0.0472 &0.0475 \\
			& & CEM(``Auto'') &$\mathcal{M}(W)$ &1.2660 &-0.2177 &0.0438 &0.2221  	 \\
			& &     & Formula (\ref{eq_6}) &1.4481 &-0.0355 &0.0524 &0.0590\\
			& & CEM(``G3'') &$\mathcal{M}(W)$ &0.7203 &-0.7634 &0.0524 &0.7652 \\
			&  &   & Formula (\ref{eq_6}) &0.7298 &-0.7539 &0.0606 &0.7563\\                	
			\\ 	                	
			\bottomrule
			\multicolumn{8}{l}{\small *$\mathcal{M}(W, W \cdot X_1, X_1, X_2)$ represents the true outcome model in Scenarios 1 and 2, but is a misspecified model in Scenarios 3 and 4.} \\
		\end{tabular}
		\label{table1} 
	\end{center}
	
\end{table}	

\bibliographystyle{plainnat}
\bibliography{references}

\end{document}